\documentclass[usletter]{article}
\usepackage{ijcai17}
\usepackage{times}
\usepackage{pgfplots}
\pgfplotsset{compat = 1.3}

\usepackage{amsmath}
\usepackage{amsthm}
\usepackage{amssymb}
\usepackage{paralist}
\usepackage{multicol}
\usepackage{amsfonts}
\usepackage{mathtools}
\usepackage{framed}
\usepackage{booktabs}
\usepackage[group-separator={,}]{siunitx}
\usepackage{multicol}
\usepackage{subcaption}
\usepackage[font=small]{caption}

\usepackage{etoolbox}
\newtoggle{isFullVersion}
\togglefalse{isFullVersion}

\usepackage{algorithm}
\usepackage[noend]{algpseudocode}

\usepackage{tikz}
\usepackage{standalone}
\usepackage{lipsum}

\usepackage{mathtools}

\newtheorem{theorem}{Theorem}

\newtheorem{definition}{Definition}
\newtheorem{problem}{Problem}

\newcommand{\cycles}{\mathcal{C}(L,0)}

\newcommand{\K}{\mathcal{K}}

\newcommand{\url}[1]{\texttt{#1}}

\newcommand*\circled[1]{%
  \tikz[baseline=(C.base)]\node[draw,circle,inner sep=0.5pt](C) {#1};\!
}

\title{Operation Frames and Clubs in Kidney Exchange}
\author{Gabriele Farina\\
Computer Science Department\\
Carnegie Mellon University\\
{gfarina@cs.cmu.edu}\\
\And John P. Dickerson\\
Computer Science Department\\
University of Maryland\\
{john@cs.umd.edu}\\
\And Tuomas Sandholm\\
Computer Science Department\\
Carnegie Mellon University\\
{sandholm@cs.cmu.edu}
}

\usetikzlibrary{arrows}
\usetikzlibrary{graphs}
\usetikzlibrary{shapes}

\usepackage{newfloat}
\DeclareFloatingEnvironment[
    fileext=lof,
    listname=List of Formulations,
    name=Formulation,
    placement=tbhp,
    within=none,
]{formulation}

\pgfplotsset{
    box plot/.style={
        /pgfplots/.cd,
        black,
        only marks,
        mark=-,
        mark size=1em,
        /pgfplots/error bars/.cd,
        y dir=plus,
        y explicit,
    },
    box plot box/.style={
        /pgfplots/error bars/draw error bar/.code 2 args={%
            \filldraw[thick,fill=white] ##1 -- ++(1em,0pt) |- ##2 -- ++(-1em,0pt) |- ##1 -- cycle;
        },
        /pgfplots/table/.cd,
        y index=2,
        y error expr={\thisrowno{3}-\thisrowno{2}},
        /pgfplots/box plot
    },
    box plot top whisker/.style={
        /pgfplots/error bars/draw error bar/.code 2 args={%
            \pgfkeysgetvalue{/pgfplots/error bars/error mark}%
            {\pgfplotserrorbarsmark}%
            \pgfkeysgetvalue{/pgfplots/error bars/error mark options}%
            {\pgfplotserrorbarsmarkopts}%
            \path[dashed] ##1 -- ##2;
        },
        /pgfplots/table/.cd,
        y index=4,
        y error expr={\thisrowno{2}-\thisrowno{4}},
        /pgfplots/box plot
    },
    box plot bottom whisker/.style={
        /pgfplots/error bars/draw error bar/.code 2 args={%
            \pgfkeysgetvalue{/pgfplots/error bars/error mark}%
            {\pgfplotserrorbarsmark}%
            \pgfkeysgetvalue{/pgfplots/error bars/error mark options}%
            {\pgfplotserrorbarsmarkopts}%
            \path[dashed] ##1 -- ##2;
        },
        /pgfplots/table/.cd,
        y index=5,
        y error expr={\thisrowno{3}-\thisrowno{5}},
        /pgfplots/box plot
    },
    box plot median/.style={
        /pgfplots/box plot
    }
}

\begin{document}

\maketitle

\begin{abstract}
A kidney exchange is a centrally-administered barter market where patients swap their willing yet incompatible donors. Modern kidney exchanges use 2-cycles, 3-cycles, and chains initiated by non-directed donors (altruists who are willing to give a kidney to anyone) as the means for swapping.
We propose significant generalizations to kidney exchange. We allow more than one donor to donate in exchange for their desired patient receiving a kidney. We also allow for the possibility of a donor willing to donate if any of a number of patients receive kidneys. Furthermore, we combine these notions and generalize them.
The generalization is to exchange among organ \emph{clubs}, where a club is willing to donate organs outside the club if and only if the club receives organs from outside the club according to given specifications.
We prove that unlike in the standard model, the uncapped clearing problem is NP-complete.
We also present the notion of operation frames that can be used to sequence the operations across batches, and present integer programming formulations for the market clearing problems for these new types of organ exchanges.
Experiments show that in the single-donation setting, operation frames improve planning by 34\%--51\%. Allowing up to two donors to donate in exchange for one kidney donated to their designated patient yields a further increase in social welfare.
\end{abstract}

\section{Introduction}\label{sec:intro}
Kidney transplantation is the most effective treatment for kidney failure. However, the demand for donor kidneys far exceeds the supply. The United Network for Organ Sharing (UNOS) reported that as of October \num{28}th, \num{2016}, the waiting list for kidney transplant had \num{99382} patients.
	
Roughly two thirds of transplanted kidneys are sourced from cadavers, while the remaining one third come from willing healthy living donors. Patients who are fortunate enough to find a willing living donor must still contend with \emph{compatibility} issues, including blood and tissue type biological compatibility. If a willing donor is incompatible with a patient, the transplantation cannot take place.

This is where \emph{kidney exchange} comes in.
A kidney exchange is a centrally-administered barter market where patients swap their willing yet incompatible donors. Modern kidney exchanges use $2$-cycles, $3$-cycles, and chains initiated by non-directed donors (altruists who are willing to give a kidney to anyone) as the means for swapping.

The idea of kidney exchange was introduced by~\citeauthor{Rapaport86:Case}~[\citeyear{Rapaport86:Case}], and the first organized kidney exchanges started around \num{2003}~\cite{Roth04:Kidney,Roth05:Pairwise}. Today there are kidney exchanges in the US, Canada, UK, the Netherlands, Australia, and many other countries. In the US, around \num{10}\% of live-donor kidney transplants now take place via exchanges.

Kidney exchanges started as matching markets where one donor-patient pair would give to, and receive from, another donor-patient pair. In other words, $2$-cycles~\cite{Roth05:Pairwise} were the structures used. Then, kidney exchange was generalized to also use $3$-cycles~\cite{Roth07:Efficient}, and then short and finally never-ending chains initiated by non-directed donors (altruists who are willing to give to anyone without needing an organ in return)~\cite{Roth06:Utilizing,Rees09:Nonsimultaneous}.

Significant work has been invested into scaling the market clearing algorithms, that is, the algorithms that find the optimal combination of non-overlapping (because any one donor can give at most one kidney) cycles and chains~\cite{Abraham07:Clearing,Constantino13:New,Manlove15:Paired,Anderson15:Finding,Dickerson16:Position}.
	
We propose a significantly generalized, more expressive, approach to kidney exchange. We allow more than one donor to donate in exchange for their desired patient receiving a kidney.
We also allow for the possibility of a donor willing to donate if any of a number of patients receive kidneys. 
Furthermore, we combine these notions and generalize them.

Our generalization can be formalized around the concept of exchange among organ \emph{clubs}, where, roughly speaking, a club is willing to donate organs outside the club if and only if the club receives organs from outside the club according to given specifications.
More specifically, exchange clubs extend the notion of a donor-pair pair, allowing for a set of healthy donors equally willing to donate one of their kidneys in exchange for an equal (or greater) number of kidneys received by a target set of patients.

Forms of organ clubs already exist---under an arrangement where one gets to be in the club as a potential recipient if one is willing to donate one's organs to the club upon death. For example, there was such a club called \emph{LifeSharers} in the US for several years~\cite{Hennessey06:Members-Only}. It shut down in \num{2016} amid controversy regarding whether an organ club would actually hurt the nationwide organ allocation.  Similarly, there is an organ club in the military ``that allows families of active-duty troops to stipulate that their loved ones' organs go to another military patient or family~\cite{Kime16:Defense}.'' Also, Israel started an organ club where those who have given consent to become organ donors upon death (or whose family members have donated an organ in the past) get priority on the organ waitlist if they need organs; this increased organ donation in Israel by 60\% in just one year~\cite{Stoler16:Incentivizing,Ofri12:Israel}.
One way to think of the approach that we are proposing is as an inter-club exchange mechanism that increases systemwide good---and can
also
be applied to live donation.

Our approach is beneficial also in a setting where there are no organ clubs in the traditional sense. We will nevertheless find the notion of a club useful in a technical sense to define the constraints, as we will detail later.
We propose a formalization of this new kind of organ exchange, and propose an organ exchange approach where clubs are conceptually the primary agents---whether they are actually clubs, altruists, or donor-patient pairs, or a combination thereof. We support both intra-club and inter-club donations.
We prove that unlike in the standard model, the uncapped clearing problem is NP-complete.

To address
the issues that (1) a club (of which a donor-patient pair is a special case, as is an altruist donor) wants to receive no later than it gives, and (2) there are logistical limits as to how many operations can be conducted simultaneously, we introduce the concept of \emph{operation frames}. They provide a convenient framework for handling the problem of synchronizing different transplants, by imposing a partial order on them. We propose a linear integer program for this.

Experiments show that in the single-donation setting, operation frames improve planning by 34\%--51\%. Allowing up to two donors to donate in exchange for one kidney donated to their designated patient yields a further welfare increase.

\section{The Standard Model}\label{sec:model}
Today's kidney exchanges (and other modern barter exchanges) can be modeled as follows. There is a directed \emph{compatibility graph} $G = (V, E)$, where vertices represent participating parties and edges representing potential transactions~\cite{Roth07:Efficient,Abraham07:Clearing}.  In the kidney exchange context, the set of vertices $V$ is partitioned as $V = V_p \cup V_n$, where $V_p$ is the set of donor-patient pairs, and $V_n$ is the set of \emph{non-directed} donors (NDDs).
	
For sake of simplicity, we will consider all non-directed donor vertices as formal donor-patient pairs, where the patient is an artificial object---denoted by $\bot$---that is incompatible with any donor in the system. Vertices $u$ and $v$ are connected by a directed edge $u \to v$ if the donor in $u$ is compatible with the patient in $v$.
The exchange administrator can also define a weight function $w : E \to \mathbb{R}$ representing, for each edge $e = (u,v) \in E$, the underlying quality or priority given to a potential transplant from $u \to v$.
	
Given the model above, we wish to solve the \emph{clearing problem}, that is, we wish to select some subset of edges with maximum total weight subject to underlying feasibility constraints.  For example, a donor $d$ in a donor-patient pair $v = (d,p) \in V_p$ will donate a kidney if and \emph{only if} a kidney is allocated to his or her paired patient $p$.  Non-directed donors have no such constraint.  In the model described so far, any solution consists of only two kinds of structure:
\begin{itemize}
\item \emph{chains}, that is paths in $G$ initiated by NDDs and then consisting entirely of donor-patient pairs; and
\item \emph{cycles}, that is loops in $G$ consisting of vertices in $V_p$---and not non-directed donors in $V_n$.
\end{itemize}
Furthermore, in any feasible solution, these structures cannot share vertices: no donor can give more than one kidney.
	
In kidney exchange, a length cap $L$ is imposed on cycles for logistical reasons.  All transplants in a cycle must be performed simultaneously so that no donor can back out after his patient has received a kidney but before he has donated his kidney. In most fielded exchanges worldwide, $L=3$, so only $2$-cycles and $3$-cycles are allowed.

Chains do not need to be constrained in length, because it is not necessary to enforce that all transplants in the chain occur simultaneously. There is a chance that a donor backs out of her commitment to donate, but this event is less catastrophic than the equivalent in cycles. Indeed, a donor backing out in a cycle results in some other patient in the pool losing his donor while not receiving a kidney---that is, a participant in the pool is \emph{strictly} worse off than before---while a donor backing out in a chain simply results in the chain ending.  While that latter case is unfortunate, no participant in the pool is strictly worse off than before. In practice, however, a chain length cap is used, in order to make the planned solution more robust to last-minute failures~\cite{Dickerson12:Optimizing,Dickerson16:Position}.

The problem can be formulated as an integer program to find the optimal solution, and indeed there has been significant work on developing increasingly scalable integer programming algorithms and formulations for this problem (e.g.,~\cite{Roth07:Efficient,Abraham07:Clearing}). The state of the art formulation is called PICEF~\cite{Dickerson16:Position}. Its number of variables is polynomial in chain length cap and exponential in cycle length cap, which is not a problem in practice because the latter cap is small. Furthermore, the LP relaxation is very tight, causing good upper bounding in the search tree and therefore fast run time.
		
\section{Exchange Clubs as a Modeling Construct}\label{sec:clubs}
We propose significant generalizations to (kidney) exchange. We allow more than one donor to donate in exchange for their desired patient receiving a kidney. We also allow for the possibility of a donor willing to donate if any of a number of patients receive kidneys. Furthermore, we combine these notions and generalize them.
We formalize this by introducing the modeling concept of \emph{exchange clubs}.
\begin{definition}{\bf (Exchange club)}
An \emph{exchange club} $c$ is a tuple $(D_c, P_c, \alpha_c, \gamma_c)$ composed of
\begin{itemize}
	\item a (possibly empty) set of donors $D_c$;
	\item a (possibly empty) set of patients $P_c$;
	\item a real $\alpha_c \ge 1$ called ``matching multiplier''. Intuitively, this means that for each matched patient in $P_c$, the club is willing to donate (in expectation) $\alpha_c$ kidneys to the pool;
	\item a real $\gamma_c \ge 0$ called ``matching debt''.
\end{itemize}
\end{definition}

The idea of exchange clubs is that donors in $D_c$ are willing to donate kidneys only if doing so results in a tangible benefit (that is, kidneys donated) to patients in $P_c$.
More precisely, let $n^\text{ext}_d(t)$ be the number of kidneys donated from donors in $D_c$ to clubs other that $c$ by time $t$, and let $n^\text{ext}_p(t)$ be the number of kidneys donated from donors outside of $c$ to patients in $P_c$; then the following inequality must hold for all time $t$ in order for club $c$ to be willing to participate in the solution:
\begin{equation}\label{eq:nd}
	n^\text{ext}_d(t) \le \alpha_c n^\text{ext}_p(t) + \gamma_c
\end{equation}
For now, we ignore parameter $\gamma_c$, whose role and motivation will become clear later.

We can now formalize the uncapped generalized clearing problem as follows.
\begin{definition}{\bf (Disjoint clubs)}
	We say that two exchange clubs $c$ and $c'$ are \emph{disjoint} if $P_c \cap P_{c'} = \varnothing$ and $D_c \cap D_{c'} = \varnothing$.
\end{definition}

\begin{problem}{\bf (Uncapped generalized clearing problem)}
	Let $\cal{C}$ be a set of mutually disjoint exchange clubs; let ${\cal D} = \cup_{c\in{\cal C}} D_c$ and ${\cal P} = \cup_{c\in{\cal C}} P_c$ denote the overall set of donors and patients respectively. Furthermore, let $E \subseteq {\cal D}\times{\cal P}$ be the set of compatibility edges, and let $w : E \to \mathbb{R}$ a weighting function assigning a weight to every compatibility edge. We want to find a set of edges that maximizes the sum of weights and satisfies Inequality~\ref{eq:nd} assuming all the selected transplants occur simultaneously.
\label{pb:uncapped}
\end{problem}

\textbf{Matching Debts}.
We now explain the meaning of $\gamma_c$. Suppose a number $n^\text{ext}_p$ of patients in club $c$ receive kidneys from other clubs, and that the optimal solution of the problem requires that $n^\text{ext}_d$ donors from club $c$ donate a kidney to other clubs. If $n^\text{ext}_d < \alpha_c n^\text{ext}_p$, we say that club $c$ owes $\alpha_c n^\text{ext}_p - n^\text{ext}_d$ kidneys to the system. This is exactly the meaning of the ``matching debt'' of a club.  It reflects the sum of all debts that a club has cumulated in the past. Except for clubs defined by non-directed donors (which start with a debt of 1), each club typically starts with a debt of 0 at the beginning, and potentially increases or decreases its debt to the system over time.

The kidney exchange pool changes over time as the exchange conducts transplants. Debts allow the exchange to keep track of the state, which is then used as the input in the next optimization.

\subsection{The Standard Model is a Special Case}
\label{se:standard_model_as_special_case}
The (uncapped) standard model is a special case of our model:
\begin{itemize}
	\item each non-directed donor defines a club $c$ with no patient, and where he or she is the only donor. Furthermore, the club has $\gamma_c = 1$ (the value of $\alpha_c$ is irrelevant);
	\item each $(d,p)$ donor-patient pair in the standard models defines a club $c$, where $D_c = \{d\}$, $P_c = \{p\}$ and $\alpha_c = 1$.
\end{itemize}

At the same time, our new model allows for some important generalizations. For instance, consider the case where one patient $p$ has a set of two donors both willing to donate a kidney in exchange for only one kidney donated to $p$. In this case, the two donors and $p$ form a club with $\alpha_c = 2$.

The introduction of exchange clubs as a modeling construct calls for a different representation of the problem because the traditional donor-patient pairs cannot capture all the new aspects. Therefore, we explicitly represent donors and patients as different types of vertices in the graph. Figure~\ref{fi:club graph example} illustrates this under the further assumption that $\alpha_c=1, \gamma_c =0$ for all clubs. We represent donor vertices with a square and patient vertices with a circle. Observe that in Figure~\ref{fi:club graph example} it is not possible to extend the given solution with an edge from Donor~8 to Patient~7, as doing so would violate Inequality~\ref{eq:nd}: Club~D does not receive any kidney from other clubs, and therefore it cannot be asked to donate.
\begin{figure}[!h]
	\centering\includegraphics[width=.8\linewidth]{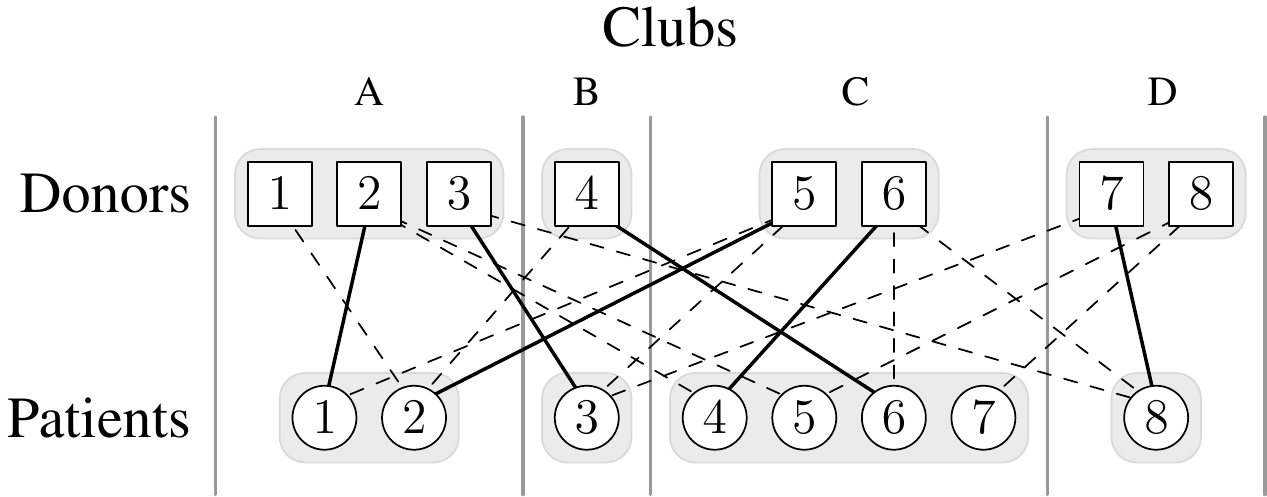}
	\caption{Tiny example problem instance. Vertical dashed lines separate different exchange clubs (so that, for instance, the first club has $D_1 = \{1, 2, 3\}, P_1 = \{1, 2\}$). Solid edges show a solution. Dashed edges represent unused compatibilities.
This figure assumes $\alpha_c = 1$, $\gamma_c = 0$ for all four clubs.}
	\label{fi:club graph example}
	\vspace{-4mm}
\end{figure}

\subsection{Uncapped Problem Formulation}\label{ssec:uncapped formulation}
It is not clear how one could apply an integer program formulation like the state-of-the-art PICEF formulation for the standard kidney exchange problem~\cite{Dickerson16:Position} in this new setting. This is because, in our setting, the feasibility (in the sense of Inequality~\ref{eq:nd}) of a particular donation depends on what transplants have already been conducted. It does not seem immediate how such aspects could be encoded in a formulation like PICEF.
However, one can easily write an integer linear program for the uncapped clearing problem that selects edges so as to maximize total weight of the selected edges subject to satisfying Constraint~\ref{eq:nd}.

Theorem~\ref{thm:hardness} shows that the decision problem associated with this problem is
is NP-complete. This is in stark contrast to the standard model where the uncapped version can be solved in polynomial time~\cite{Abraham07:Clearing}. This increase in hardness is the cost of our increased expressiveness.
\begin{theorem}\label{thm:hardness}
The decision problem associated with the uncapped generalized clearing problem is NP-complete.
\end{theorem}
\begin{proof}[Proof sketch.]
   \noindent\emph{Membership in NP}: Given a set of mutually disjoint exchange clubs $\mathcal{C}$ and set of $k$ trades, it is trival to check in polynomial time if they satisfy Inequality~\ref{eq:nd}.

   \noindent\emph{NP-hardness}: We reduce from \textsc{SET-PACKING (SP)}.  An instance of \textsc{SP} takes a set of items $\mathcal{U}$, a family $\mathcal{S}$ of subsets of $\mathcal{U}$, and an integer $k$ as input; the task is to find a disjoint subfamily $\mathcal{X} \subseteq \mathcal{S}$ such that $|\mathcal{X}| = k$.

Assume that we are given an instance of \textsc{SP}.  We will now build an instance of our problem.  Let $n = |\mathcal{U}|$ be the number of items and $m = |\mathcal{S}|$ be the number of subsets.  Index the items $\{u_1, \ldots, u_n\} \in \mathcal{U}$ and the subsets $\{S_1, \ldots, S_m\} \in \mathcal{S}$.
Construct a disjoint set of clubs $\mathcal{C}$ as follows.  For each $u_i \in \mathcal{U}$, construct a club $a_i$ with no patient, one donor, and $\gamma_{a_i} = 1$.  For each subset $S_j \in \mathcal{S}$, construct a club $c_j$ with one patient and no donor.  Furthermore, for each subset $S_j$, construct a club $b_j$ with one donor and $\ell = |S_j|$ patients, $\gamma_{b_j} = 0$, and $\alpha_{b_j} = 1/\ell$.  Intuitively, this club will donate its one kidney iff each of the $\ell$ patients receives a kidney.

We now specify the set of legal transplants.  Let $M = |\mathcal{U}| + 1$.  For each subset $S_j \in \mathcal{S}$, draw a directed edge with weight $M$ from the single donor in club $b_j$ to the single patient in club $c_j$.  Furthermore, for each item $u_i \in \mathcal{U}$ and subset $S_j \in \mathcal{S}$ such that $u_i \in S_j$, draw one directed edge with weight $1$ from the single donor in club $a_i$ to the patient corresponding to item $u_i$ in club $b_j$.  Figure~\ref{fig:hardness} shows the final construction.
\begin{figure}[ht!bp]
\begin{subfigure}[b]{.38\linewidth}
\begin{center}
{\small Universe set:}\\
$\{1,2,3,4,5,6,7\}$\\
\vspace{2mm}
\small
{Set packing instance:}\\
\begin{tabular}{cc}
  \toprule
  $S_1$ & $\{1,3,4\}$\\
  $S_2$ & $\{1,3,5,7\}$\\
  $S_3$ & $\{4,6\}$\\
  \bottomrule
\end{tabular}
\end{center}
\caption{}
\end{subfigure}
\hfill
\begin{subfigure}[b]{.60\linewidth}
\centering
\includegraphics[width=.92\linewidth]{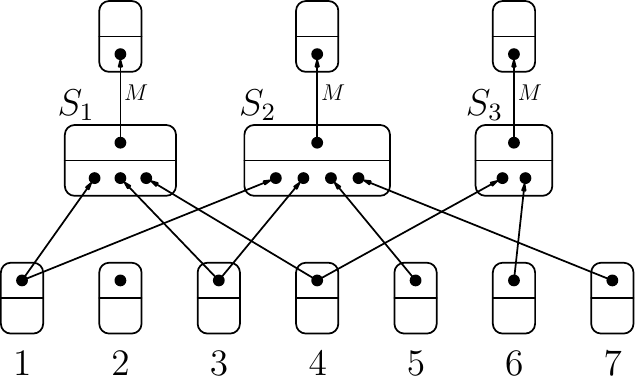}
\caption{}
\end{subfigure}
\caption{(b) Constructed instance of our problem. Each nodes represents a club, with donors in the top part, and patients in the bottom part. Edges represent compatibilities.}
\vspace{-.1in}
\label{fig:hardness}
\end{figure}

We will now show that a solution exists for the instance of \textsc{SC} if and only if our problem has a legal matching with weight in $[kM, (k+1)M)$.

\noindent
($\Rightarrow$)  Suppose there exists some solution $\mathcal{S'} = \{S'_1, \ldots, S'_k\}$ to the \textsc{SP} problem; that is, there exists some disjoint subfamily of $\mathcal{S}$ of size $k$.  Then, for each subset $S'_j \in \mathcal{S'}$, for each element $u_i$ in $S'_j$, use the edge from club $a_i$ to club $b_j$.  By the disjointness of the subfamily $\mathcal{S'}$, each single-donor club $a_i$ for $i \in [n]$ donates to at most one club.  Furthermore, each club $b_j$ corresponding to $S'_j \in \mathcal{S}'$ receives one kidney for each patient in its club; by construction, their matching multiplier is now satisfied.  Thus, for each $S'_j \in \mathcal{S'}$, include the edge with weight $M$ from the one donor in $b_j$ to the one patient in $c_j$.  This results in a matching of weight at least $kM$, but no more than $kM + n < (k+1)M$.

\noindent
($\Leftarrow$)  Suppose there exists a matching in our problem such that the weight of the matching is in $[kM, (k+1)M)$.  Then exactly $k$ of the edges between exactly $k$ pairs of clubs $b_\cdot$ and $c_\cdot$ are used, at total weight $kM$.  Let $j' \in [m']$ index those clubs $b_{j'}$ that use their one outgoing edge to club $c_{j'}$.  Each club $b_{j'}$ uses that edge if and only if every one of its internal patients receives a kidney; since each single-donor club $a_\cdot$ can give at most one kidney, they are used at most once.  Since at most $n$ clubs $a_\cdot$ can be used, each at additional weight $1$, the final matching is of weight at most $km+n < (k+1)M$; further, exactly $k$ clubs $b_{j'}$ were fulfilled completely, corresponding to exactly $k$ disjoint subsets $S_{j'} \in \mathcal{S}$ being packed.
\end{proof}

\section{Operation Frames}

While the above modeling approach is promising, it has a major shortcoming: it might require that a potentially large number of operations happen at the same time so as to honor the condition that the donors not be operated on strictly before patients in their clubs receive kidneys. An example is provided in Figure~\ref{fi:cycle}, where we would need patients $\{1,2,3,6,7\}$ and donors $\{1,2,3,5,6\}$ to be operated on at the same time.
\begin{figure}[!h]
\vspace{-2mm}
\centering\includegraphics[width=.8\linewidth]{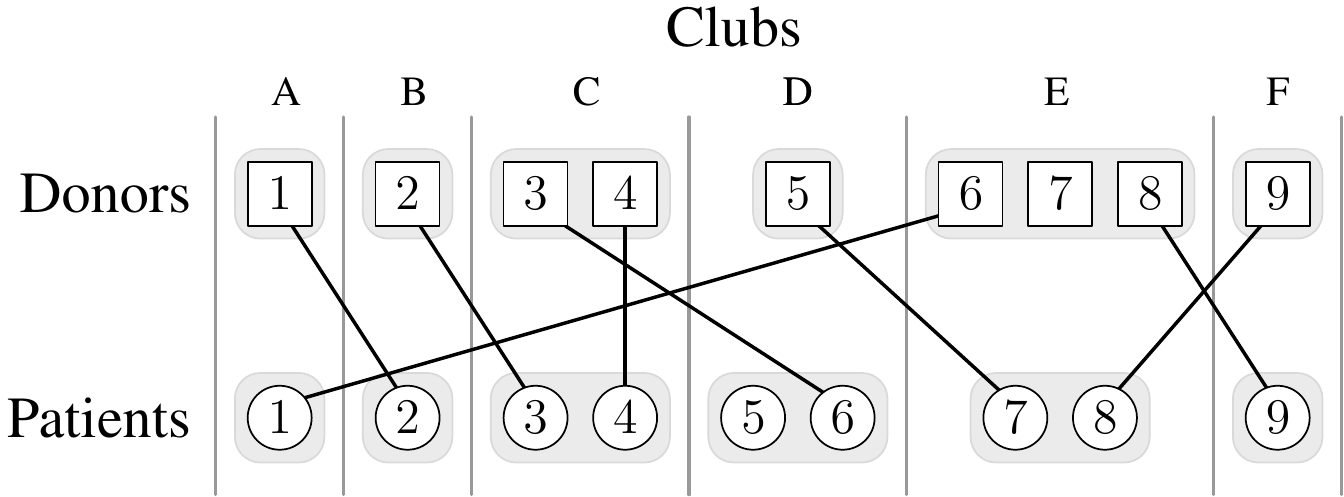}
	\caption{Example that requires $5$ simultaneous transplants.}
\vspace{-2mm}
\label{fi:cycle}
\end{figure}

\noindent This is not practically viable for at least two reasons:
\begin{itemize}
\item the success probability of all the planned transplants in the structure succeeding in their pre-operation blood type compatibility tests (aka. crossmatch test) and other pre-transplant testing decreases multiplicatively with the number of edges in the planned structure,\footnote{For further details about pre-transplant test failures, see Dickerson, Procaccia, and Sandholm~\shortcite{Dickerson13:Failure} and Blum et al.~\shortcite{Blum15:Ignorance}.} and
\item the logistic (and financial) details are hard to execute---ten people to operate on have to be coordinated, together with the surgeons and staff needed for ten surgeries.
\end{itemize}

\noindent In order to solve this synchronization issue, we introduce the concept of \emph{operation frames}. An operation frame $t$ is an edge set of size up to $K_t$, representing operations to be performed at the same time.
The introduction of operation frames enables us to reason in terms of order in which the operations will be carried out. The chronological order imposed on the operation frames is \emph{partial}. For this reason, we can formalize the set and relationships among operation frames by means of a directed acyclic graph (DAG) $F=(T, B)$, where the set of vertices (i.e., $T$) coincides with the set of operation frames, while the set of edges $B \subseteq T\times T$ denotes the \emph{happens-strictly-before} chronological (partial) order. We say that operation frame $u$ happens strictly before operation frame $v$, denoted $u \hspace{-1pt}\rightsquigarrow\hspace{-1pt} v$, if there exists a directed path in $F$ from $u$ to $v$.

The introduction of operation frames enables Inequality~\ref{eq:nd} to be written in terms of \emph{logical} time, that is, substituting the notion of time with the partial happens-strictly-before order.
Thus, for any frame $\tau \in T$, the number of kidneys that were surely (i.e., for any possible linearization of the DAG $F$) donated to and from club $c$ at the time when $\tau$ is executed is
$
	n_d(\tau) = \sum_{\tau' \rightsquigarrow \tau} d(c, \tau'),\quad
	n_p(\tau) = \sum_{\tau' \rightsquigarrow \tau} p(c, \tau'),\nolinebreak
$
where $d(c, \tau')$ and $p(c, \tau')$ represent the number of transplants from and to club $c$ scheduled for operation frame $\tau'$.

We argue that operation frames provide a richer problem structure, as it is now possible to assign a (partial) chronological order to the operations we plan to perform. Furthermore, they encode the condition that ``no more than $K_t$ people get operated on at the same time'' in a very natural way: every operating frame has a parameter $K_t$. Indeed, operation frames guarantee that not too many surgeries are planned to happen at the same time. Conceptually, they are equivalent to imposing chain and cycle length caps, as it is done in standard model. However, here we are allowing much richer exchange structures (and, as presented so far, there is no way of specifying a different cap on the size of chains versus cycles versus other structures). This is because the operation frame size cap $K_t$ represents an actual limit on the number of simultaneous operations that can be accommodated. Different frames can have different size constraints.

Finally, operation frames allow an integer programming formulation of the problem that (unlike PICEF) uses a number of variables that is polynomial in the maximum size cap $\max_t K_t$. We present that formulation in the next section.

\subsection{Capped Problem Formulation}\label{ssec:capped formulation}
The idea of operation frames can be plugged into the uncapped integer program,
leading to the following formulation.

\begin{formulation}[!h]
\centering\includegraphics[width=.82\linewidth]{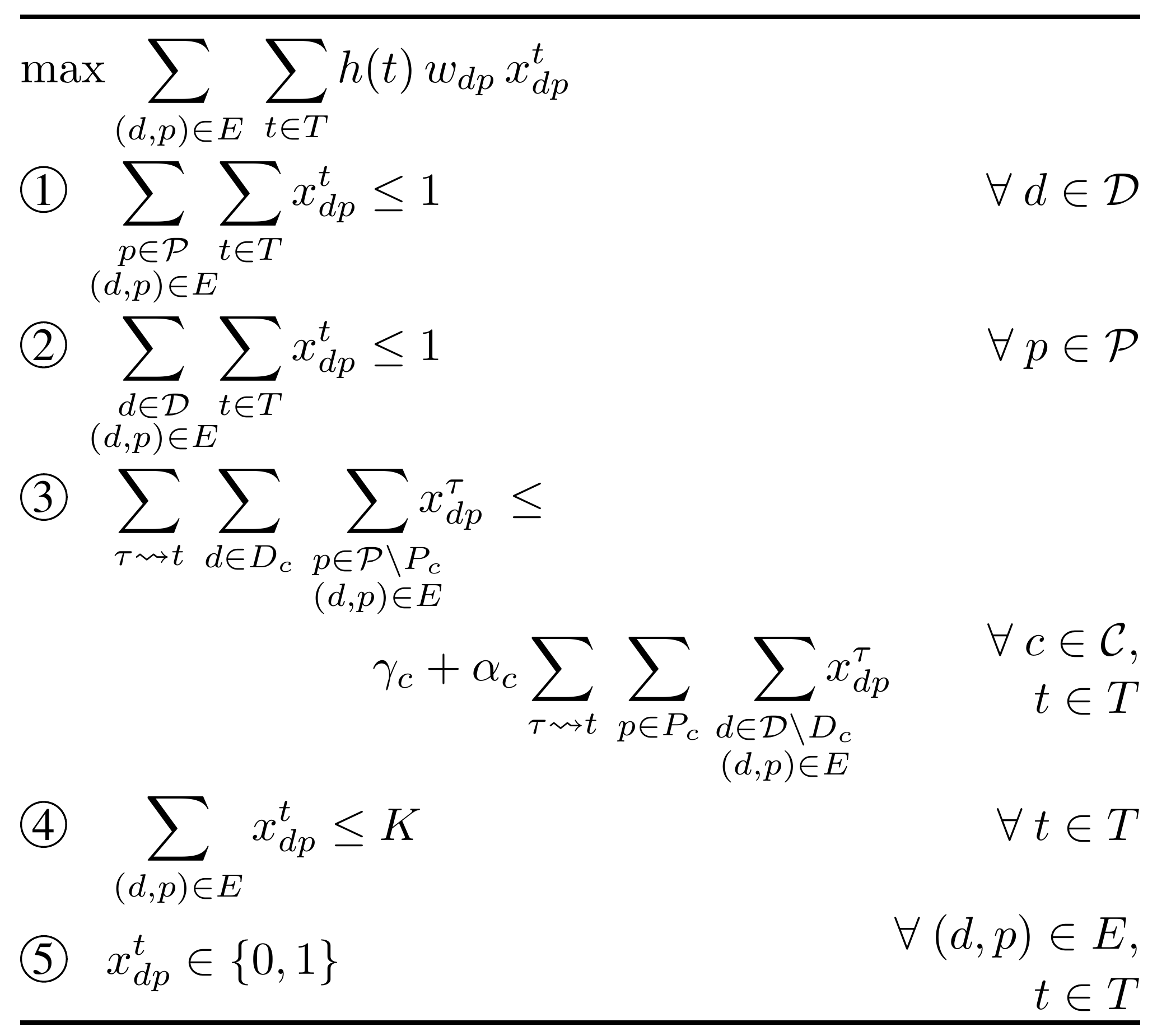}

	\caption{MIP formulation for the capped problem.}
	\label{fo:capped}
\end{formulation}

We let $x^t_{dp}$ be a binary value (Constraint~\circled{5}) indicating whether the transplant represented by the edge $(d,p)\in E$ is scheduled for operation frame $t$. Analogous to the uncapped case, Constraints~\circled{1} and~\circled{2} ensure that each donor donates at most one kidney and that each patient receives at most one kidney, respectively. Constraint~\circled{4} ensures that in any operation frame $t\in\cal T$, no more than $K$ transplants are  scheduled. Constraint~\circled{3} enforces that at any time $t$, for each club $c \in\cal C$, the total number of kidneys donated from club $c$ does not exceed $\lfloor \gamma_c + \alpha_c n^{\text{ext}}_p(t) \rfloor$, where $n^{\text{ext}}_p(t)$ is the total number of kidneys donated to club $c$ from other clubs, before or at operation frame $t$.
The objective function ensures that a maximum-weight solution is found.\footnote{One can also model temporal preferences by multiplying the edge weights by discounts $h(t)$, which depend on which operation frame $t$ the surgery is conducted. (This assumes that the time between frames is exogenous---but not necessarily constant---that is, the time between frames does not depend on what transplants the optimizer decides to put in each frame.) This discounting is already include in the objective in Formulation~\ref{fo:capped}.}

\subsection{Operation Frames Counter Myopia}
Present-day kidney exchanges operate in a batch setting, potentially planning in a single shot long chains that will, in practice, execute in segments over many months.  Solvers for the standard problem (e.g., those based on PICEF) optimize on a batch-by-batch basis, selecting the global optimum solution only inside of a single batch, and not considering future batches.  Our approach is more powerful than the standard batch-based one also in the sense that it inherently breaks long structures into shorter ones that execute sequentially. Figure~\ref{fi:myopic} shows an example compatibility graph where our model will return a higher-value solution than the optimal solution in the traditional model.
\begin{figure}[!h]
  \vspace{-3mm}
  \centering
  \includegraphics[width=0.75\linewidth]{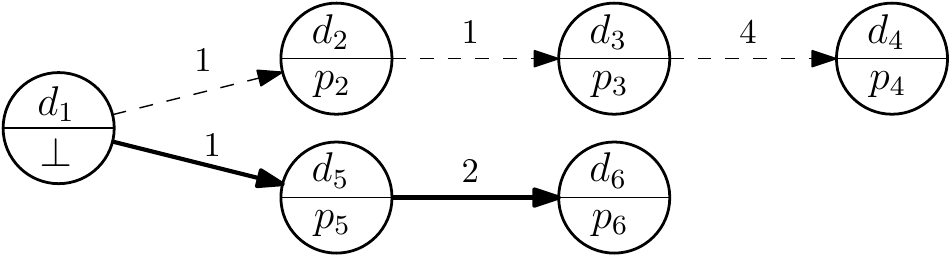}
  \vspace{-2mm}
  \caption{An example graph where Formulation~\ref{fo:capped} returns a higher-value solution than the standard batch approach, even with only the standard kinds of vertices available.  Given a cap $K=2$, the standard approach will choose the lower $2$-chain over the upper $2$-chain for utility $3$, while our solver will choose to match the upper $3$-chain across two frames for greater overall utility of $6$.}
  \label{fi:myopic}
  \vspace{-2mm}
\end{figure}

Figure~\ref{fo:capped} shows that the capped approach in the standard model cannot consider certain solutions; yet, our capped approach---based on the concept of operations frames---optimizes across all the operation frames at the same time, resulting in less myopic behavior. Under the assumption that the kidney exchange pool is not affected by any exogenous behavior (e.g., compatibility failures, deaths of donors or patients, dynamic insertions and deletions of edges and vertices), our formulation is guaranteed to find a globally optimal allocation of transplants across all operation frames. In contrast, even under these strong assumptions, present-day solvers for the standard model will (by design) fail to find a globally optimal solution across batches.

\section{Experiments}
We conducted experiments to evaluate the techniques. The experiments are conducted using the real data from the UNOS kidney exchange that started in 2010. In that data we have the vertices that have a patient and one or more donors, and vertices that represent altruists. We also have patient and donor attributes, and we use the UNOS rules for determining compatibility of potential transplants (these rules take into account blood type, tissue type, creatinine, acceptable transplant centers, etc). So, in this master data set we effectively have a graph that has as its vertices all the vertices that have been in the UNOS exchange. All edges have unit weight.

Also, around 70\% of edges in kidney exchanges fail due to various pre-transplant tests, so the planned transplant cannot occur after all~\cite{Dickerson13:Failure}. We take this into account by randomly removing 70\% of the edges independently from the master graph.
In each experiment, we subsample vertices from this master graph (and keep the edges between the selected vertices) to generate multiple problem instances, i.e., input graphs. For each value of the pool size, we repeat the experiment with 50 different random seeds. The operation frames DAG was chosen to be a total order.

The first experiment compares the standard model (Section~\ref{se:standard_model_as_special_case}) to our capped operation-frame-based model. In both cases, at most one donor from each pair is used. In both cases, the cap is four; this is a conservative experimental design because often in the standard model cycles are further capped to be at most length three, which would disadvantage the standard model compared to our approach. In both cases, we sample around 5\% of the vertices to be altruists from the master graph and the rest are sampled uniformly from the non-altruist vertices.
Figure~\ref{fi:improvement_over_standard} shows that our approach yields a 34\%--51\% improvement via better---less myopic---planning even in this static setting.
\pgfplotsset{grid style={dotted,black!50!white}}
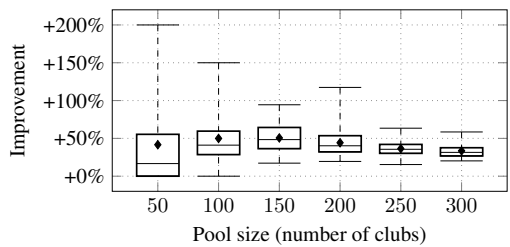
\begin{figure}[!h]
\vspace{-1mm}
\centering
\begin{tikzpicture}[scale=0.80]
\begin{axis} [grid,ylabel shift=0mm,enlarge x limits=0.15,xtick=data, width=.95\columnwidth, height=4.6cm, xlabel={Pool size (number of clubs)}, ylabel={Improvement}, ytick={0, 50, 100, 150, 200}, yticklabels={+0\%,+50\%,+100\%,+150\%,+200\%}]
    \addplot [box plot median] table {stdmaxdon1.dat};
    \addplot [box plot box] table {stdmaxdon1.dat};
    \addplot [box plot top whisker] table {stdmaxdon1.dat};
    \addplot [box plot bottom whisker] table {stdmaxdon1.dat};
    \addplot [mark=diamond*,only marks] coordinates {
      (50.000000, 41.609518)
      (100.000000, 49.919076)
      (150.000000, 50.583188)
      (200.000000, 44.211556)
      (250.000000, 36.451854)
      (300.000000, 33.500494)
    };
\end{axis}
\end{tikzpicture}
\vspace{-2mm}
\caption{Improvement in the number of matches from the capped operation frame approach over the standard approach. Diamond marks denote averages, the horizontal line in the box the median, the box edges the first and third quartiles, and the whiskers the minimum and maximum across the points.}
\vspace{-2mm}
\label{fi:improvement_over_standard}
\end{figure}

In the UNOS pool, 5.7\% of the non-altruist vertices have multiple willing donors, but at most one from each vertex will be used. The second experiment is exactly like the first, except that in the capped operation-frame-based model, up to two of the donors from each vertex can be used (i.e., $\alpha_c = 2$ for clubs $c$ that have more than one donor). Figure~\ref{fi:improvement_over_standard_with_up_to_2} shows that this leads to a larger increase over the standard approach.
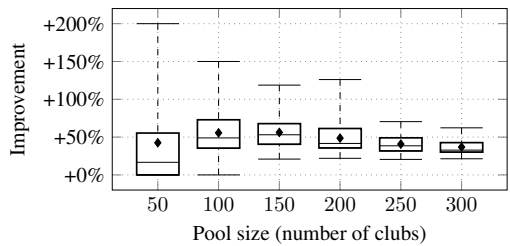
\begin{figure}[!h]
\vspace{-1mm}
\centering
\begin{tikzpicture}[scale=0.80]
\begin{axis} [grid,ylabel shift=0mm,enlarge x limits=0.15,xtick=data, width=.95\columnwidth, height=4.6cm, xlabel={Pool size (number of clubs)}, ylabel={Improvement}, ytick={0, 50, 100, 150, 200}, yticklabels={+0\%,+50\%,+100\%,+150\%,+200\%}]
    \addplot [box plot median] table {stdmaxdon2.dat};
    \addplot [box plot box] table {stdmaxdon2.dat};
    \addplot [box plot top whisker] table {stdmaxdon2.dat};
    \addplot [box plot bottom whisker] table {stdmaxdon2.dat};
    \addplot [mark=diamond*,only marks] coordinates {
		(50.000000, 42.509518)
		(100.000000, 55.587833)
		(150.000000, 56.207378)
		(200.000000, 48.643752)
		(250.000000, 40.732958)
		(300.000000, 36.815120)
    };
\end{axis}
\end{tikzpicture}
\vspace{-2mm}
\caption{Improvement in the number of matches from the capped operation-frame approach over the standard approach when up to two donors can be used from each non-altruist vertex in the former approach.}
\vspace{-3mm}
\label{fi:improvement_over_standard_with_up_to_2}
\end{figure}

In the third experiment we sample more multi-donor pairs---roughly 10\% of the entire pool. We do this to test how the system would perform if more multi-donor pairs would be present---as is conceivable in the future as knowledge about kidney exchange spreads and possibly also as multiple donors could actually be used. Figure~\ref{fi:10_percent_multidonor} shows that allowing for up to two donors to be used from a vertex leads to an improvement over allowing only up to one to be used. In both cases, we use our capped operation-frame approach.
\begin{figure}[!h]
\vspace{-1mm}
\centering
\begin{tikzpicture}[scale=0.80]
\begin{axis} [grid,ymax=25,ylabel shift=0mm,enlarge x limits=0.15,xtick=data, width=.95\columnwidth, height=4.6cm, xlabel={Pool size (number of clubs)}, ylabel={Relative improvement}, ytick={0, 5, 10, 15, 20}, yticklabels={+0\%,+5\%,+10\%,+15\%,+20\%}]
    \addplot [box plot median] table {boost.dat};
    \addplot [box plot box] table {boost.dat};
    \addplot [box plot top whisker] table {boost.dat};
    \addplot [box plot bottom whisker] table {boost.dat};
    \addplot [mark=diamond*,only marks] coordinates {
		(50.000000, 5.929654)
		(100.000000, 10.769957)
		(150.000000, 7.336270)
		(200.000000, 5.464621)
		(250.000000, 5.725581)
		(300.000000, 4.528407)
    };
\end{axis}
\end{tikzpicture}
\vspace{-2mm}
\caption{Improvement in the number of matches from allowing up to two donors to be used from each vertex compared to allowing only one donor from each vertex to be used.}
\vspace{-4mm}
\label{fi:10_percent_multidonor}
\end{figure}
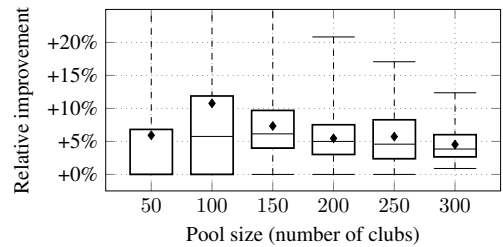

\section{Conclusions and Future Research}
Motivated by the reality of fielded kidney exchanges, in this paper we proposed significant generalizations to kidney exchange---and barter markets more generally. Specifically, we moved the model from individual and independent patient-donor pairs to the modeling concept of multi-donor and multi-patient organ \emph{clubs}, where a club is willing to donate organs outside the club if and only if the club receives organs from outside the club according to expressed preferences.  We proved that unlike in the standard model, the uncapped clearing problem is NP-complete.

We presented the notion of operation frames that sequence the operations across batches, and gave IP formulations that optimally clear these new types of markets.
Experiments show that in the single-donation setting, operation frames improve planning by 34\%--51\%. Allowing up to two donors to donate in exchange for one kidney donated to their designated patient yields a further increase in social welfare.

Operation frames include a notion of time via the \emph{happens-before} ordering; yet, this does not capture the full dynamics of kidney exchange, where vertices and edges arrive and disappear over time. Finding an optimal matching policy for fully dynamic kidney exchange is an open problem from both the theoretical~\cite{Unver10:Dynamic,Akbarpour14:Dynamic,Anderson14:Stochastic,Anderson15:Dynamic} and computational~\cite{Awasthi09:Online,Dickerson12:Dynamic,Dickerson13:Failure,Dickerson15:FutureMatch,Glorie15:Robust} points of view. Perhaps operation frames---with or without clubs---can be used to enhance planning even in those contexts as we showed it can in the static setting.

Exploring incentive issues in this new model is also interesting.
Generalizations to the basic kidney exchange model have already enabled mechanisms that circumvent strong impossibility results~\cite{Hajaj15:Strategy-Proof,Ashlagi14:Free,Ashlagi15:Mix,Toulis11:Random}.  The more expressive models presented in this paper could similarly result in advances in mechanism design.

\section*{Acknowledgments}
This work was supported by NSF grants IIS-1617590, IIS-1320620, IIS-1546752, and ARO awards W911NF-17-1-0082, W911NF-16-1-0061.

\clearpage\newpage
{\small
\bibliographystyle{named}
\bibliography{dairefs,citations}

\begin{thebibliography}{}

\bibitem[\protect\citeauthoryear{Abraham \bgroup \em et al.\egroup
  }{2007}]{Abraham07:Clearing}
David Abraham, Avrim Blum, and Tuomas Sandholm.
\newblock Clearing algorithms for barter exchange markets: Enabling nationwide
  kidney exchanges.
\newblock In {\em Proceedings of the ACM Conference on Electronic Commerce
  (EC)}, pages 295--304, 2007.

\bibitem[\protect\citeauthoryear{Akbarpour \bgroup \em et al.\egroup
  }{2014}]{Akbarpour14:Dynamic}
Mohammad Akbarpour, Shengwu Li, and Shayan~Oveis Gharan.
\newblock Dynamic matching market design.
\newblock In {\em Proceedings of the ACM Conference on Economics and
  Computation (EC)}, page 355, 2014.

\bibitem[\protect\citeauthoryear{Anderson \bgroup \em et al.\egroup
  }{2015a}]{Anderson15:Dynamic}
Ross Anderson, Itai Ashlagi, David Gamarnik, and Yash Kanoria.
\newblock A dynamic model of barter exchange.
\newblock In {\em Annual ACM-SIAM Symposium on Discrete Algorithms (SODA)},
  pages 1925--1933, 2015.

\bibitem[\protect\citeauthoryear{Anderson \bgroup \em et al.\egroup
  }{2015b}]{Anderson15:Finding}
Ross Anderson, Itai Ashlagi, David Gamarnik, and Alvin~E Roth.
\newblock Finding long chains in kidney exchange using the traveling salesman
  problem.
\newblock {\em Proceedings of the National Academy of Sciences},
  112(3):663--668, 2015.

\bibitem[\protect\citeauthoryear{Anderson}{2014}]{Anderson14:Stochastic}
Ross Anderson.
\newblock {\em Stochastic models and data driven simulations for healthcare
  operations}.
\newblock PhD thesis, Massachusetts Institute of Technology, 2014.

\bibitem[\protect\citeauthoryear{Ashlagi and Roth}{2014}]{Ashlagi14:Free}
Itai Ashlagi and Alvin~E Roth.
\newblock Free riding and participation in large scale, multi-hospital kidney
  exchange.
\newblock {\em Theoretical Economics}, 9(3):817--863, 2014.

\bibitem[\protect\citeauthoryear{Ashlagi \bgroup \em et al.\egroup
  }{2015}]{Ashlagi15:Mix}
Itai Ashlagi, Felix Fischer, Ian~A Kash, and Ariel~D Procaccia.
\newblock Mix and match: A strategyproof mechanism for multi-hospital kidney
  exchange.
\newblock {\em Games and Economic Behavior}, 91:284--296, 2015.

\bibitem[\protect\citeauthoryear{Awasthi and Sandholm}{2009}]{Awasthi09:Online}
Pranjal Awasthi and Tuomas Sandholm.
\newblock Online stochastic optimization in the large: {A}pplication to kidney
  exchange.
\newblock In {\em Proceedings of the 21st International Joint Conference on
  Artificial Intelligence (IJCAI)}, pages 405--411, 2009.

\bibitem[\protect\citeauthoryear{Blum \bgroup \em et al.\egroup
  }{2015}]{Blum15:Ignorance}
Avrim Blum, John~P. Dickerson, Nika Haghtalab, Ariel~D. Procaccia, Tuomas
  Sandholm, and Ankit Sharma.
\newblock Ignorance is almost bliss: Near-optimal stochastic matching with few
  queries.
\newblock In {\em Proceedings of the ACM Conference on Economics and
  Computation (EC)}, pages 325--342, 2015.

\bibitem[\protect\citeauthoryear{Constantino \bgroup \em et al.\egroup
  }{2013}]{Constantino13:New}
Miguel Constantino, Xenia Klimentova, Ana Viana, and Abdur Rais.
\newblock New insights on integer-programming models for the kidney exchange
  problem.
\newblock {\em European Journal of Operational Research}, 231(1):57--68, 2013.

\bibitem[\protect\citeauthoryear{Dickerson and
  Sandholm}{2015}]{Dickerson15:FutureMatch}
John~P. Dickerson and Tuomas Sandholm.
\newblock {FutureMatch}: Combining human value judgments and machine learning
  to match in dynamic environments.
\newblock In {\em AAAI Conference on Artificial Intelligence (AAAI)}, pages
  622--628, 2015.

\bibitem[\protect\citeauthoryear{Dickerson \bgroup \em et al.\egroup
  }{2012a}]{Dickerson12:Dynamic}
John~P. Dickerson, Ariel~D. Procaccia, and Tuomas Sandholm.
\newblock Dynamic matching via weighted myopia with application to kidney
  exchange.
\newblock In {\em AAAI Conference on Artificial Intelligence (AAAI)}, pages
  1340--1346, 2012.

\bibitem[\protect\citeauthoryear{Dickerson \bgroup \em et al.\egroup
  }{2012b}]{Dickerson12:Optimizing}
John~P. Dickerson, Ariel~D. Procaccia, and Tuomas Sandholm.
\newblock Optimizing kidney exchange with transplant chains: Theory and
  reality.
\newblock In {\em International Conference on Autonomous Agents and Multi-Agent
  Systems (AAMAS)}, pages 711--718, 2012.

\bibitem[\protect\citeauthoryear{Dickerson \bgroup \em et al.\egroup
  }{2013}]{Dickerson13:Failure}
John~P. Dickerson, Ariel~D. Procaccia, and Tuomas Sandholm.
\newblock Failure-aware kidney exchange.
\newblock In {\em Proceedings of the ACM Conference on Electronic Commerce
  (EC)}, pages 323--340, 2013.

\bibitem[\protect\citeauthoryear{Dickerson \bgroup \em et al.\egroup
  }{2016}]{Dickerson16:Position}
John~P. Dickerson, David Manlove, Benjamin Plaut, Tuomas Sandholm, and James
  Trimble.
\newblock Position-indexed formulations for kidney exchange.
\newblock In {\em Proceedings of the ACM Conference on Economics and
  Computation (EC)}, 2016.

\bibitem[\protect\citeauthoryear{Glorie \bgroup \em et al.\egroup
  }{2015}]{Glorie15:Robust}
Kristiaan Glorie, Margarida Carvalho, Miguel Constantino, Paul Bouman, and Ana
  Viana.
\newblock Robust models for the kidney exchange problem, 2015.
\newblock Working paper.

\bibitem[\protect\citeauthoryear{Hajaj \bgroup \em et al.\egroup
  }{2015}]{Hajaj15:Strategy-Proof}
Chen Hajaj, John~P. Dickerson, Avinatan Hassidim, Tuomas Sandholm, and David
  Sarne.
\newblock Strategy-proof and efficient kidney exchange using a credit
  mechanism.
\newblock In {\em AAAI Conference on Artificial Intelligence (AAAI)}, pages
  921--928, 2015.

\bibitem[\protect\citeauthoryear{Hennessey}{2006}]{Hennessey06:Members-Only}
Jamie Hennessey.
\newblock A members-only club for organ donors.
\newblock ABC News website, 2006.
\newblock \small\url{https://goo.gl/ZzT8eF}.

\bibitem[\protect\citeauthoryear{Kime}{2016}]{Kime16:Defense}
Patricia Kime.
\newblock {D}efense {D}epartment charts new territory in kidney transplants.
\newblock Military Times website, 2016.
\newblock \small\url{https://goo.gl/Wwe9Hh}.

\bibitem[\protect\citeauthoryear{Manlove and
  {O'M}alley}{2015}]{Manlove15:Paired}
David Manlove and Gregg {O'M}alley.
\newblock Paired and altruistic kidney donation in the {UK}: Algorithms and
  experimentation.
\newblock {\em {ACM} Journal of Experimental Algorithmics}, 19(1), 2015.

\bibitem[\protect\citeauthoryear{Ofri}{2012}]{Ofri12:Israel}
Danielle Ofri.
\newblock In israel, a new approach to organ donation.
\newblock The New York Times website, 2012.
\newblock \small\url{https://goo.gl/J37J9e}.

\bibitem[\protect\citeauthoryear{Rapaport}{1986}]{Rapaport86:Case}
F.~T. Rapaport.
\newblock The case for a living emotionally related international kidney donor
  exchange registry.
\newblock {\em Transplantation Proceedings}, 18:5--9, 1986.

\bibitem[\protect\citeauthoryear{Rees \bgroup \em et al.\egroup
  }{2009}]{Rees09:Nonsimultaneous}
Michael Rees, Jonathan Kopke, Ronald Pelletier, Dorry Segev, Matthew Rutter,
  Alfredo Fabrega, Jeffrey Rogers, Oleh Pankewycz, Janet Hiller, Alvin Roth,
  Tuomas Sandholm, Utku {\"{U}}nver, and Robert Montgomery.
\newblock A nonsimultaneous, extended, altruistic-donor chain.
\newblock {\em New England Journal of Medicine}, 360(11):1096--1101, 2009.

\bibitem[\protect\citeauthoryear{Roth \bgroup \em et al.\egroup
  }{2004}]{Roth04:Kidney}
Alvin Roth, Tayfun S{\"{o}}nmez, and Utku {\"U}nver.
\newblock Kidney exchange.
\newblock {\em Quarterly Journal of Economics}, 119(2):457--488, 2004.

\bibitem[\protect\citeauthoryear{Roth \bgroup \em et al.\egroup
  }{2005}]{Roth05:Pairwise}
Alvin Roth, Tayfun S{\"{o}}nmez, and Utku {\"U}nver.
\newblock Pairwise kidney exchange.
\newblock {\em Journal of Economic Theory}, 125(2):151--188, 2005.

\bibitem[\protect\citeauthoryear{Roth \bgroup \em et al.\egroup
  }{2006}]{Roth06:Utilizing}
Alvin Roth, Tayfun S{\"{o}}nmez, Utku {\"U}nver, Frank Delmonico, and Susan~L.
  Saidman.
\newblock Utilizing list exchange and nondirected donation through `chain'
  paired kidney donations.
\newblock {\em American Journal of Transplantation}, 6:2694--2705, 2006.

\bibitem[\protect\citeauthoryear{Roth \bgroup \em et al.\egroup
  }{2007}]{Roth07:Efficient}
Alvin Roth, Tayfun S{\"{o}}nmez, and Utku {\"U}nver.
\newblock Efficient kidney exchange: {C}oincidence of wants in a market with
  compatibility-based preferences.
\newblock {\em American Economic Review}, 97:828--851, 2007.

\bibitem[\protect\citeauthoryear{Stoler \bgroup \em et al.\egroup
  }{2016}]{Stoler16:Incentivizing}
A~Stoler, J~Kessler, T~Ashkenazi, A~Roth, and J~Lavee.
\newblock Incentivizing authorization for deceased organ donation with organ
  allocation priority: The first 5 years.
\newblock {\em American Journal of Transplantation}, 16(9):2639--45, 2016.

\bibitem[\protect\citeauthoryear{Toulis and Parkes}{2011}]{Toulis11:Random}
Panos Toulis and David~C. Parkes.
\newblock A random graph model of kidney exchanges: efficiency,
  individual-rationality and incentives.
\newblock In {\em Proceedings of the ACM Conference on Electronic Commerce
  (EC)}, pages 323--332. ACM, 2011.

\bibitem[\protect\citeauthoryear{{\"U}nver}{2010}]{Unver10:Dynamic}
Utku {\"U}nver.
\newblock Dynamic kidney exchange.
\newblock {\em Review of Economic Studies}, 77(1):372--414, 2010.

\end{thebibliography}
}

\appendix

\section{Appendix A: Limited Horizon Approximation}
Fielded kidney exchanges routinely solve NP-hard problems when matching patients to donors.  The models we presented here, being generalizations of the standard kidney exchange model, are also solving an NP-hard problem.  We presented MIP formulations to solve these models, but---by virtue of being intractable from a complexity theory point of view---it is the case that certain inputs exist that will make these formulations run slowly.  In the event that MIP Formulation~\ref{fo:capped} takes too much time, it is possible to consider the following algorithm:
\begin{enumerate}
	\item Solve the problem above with ${\cal T} = \{1, \dots, \tilde T\}$, where $\tilde T$ is small;
	\item Remove the matched structures, update the matching debts $\gamma_c$ for all clubs;
	\item Repeat steps 1-2 until no more matches are found.
\end{enumerate}
This is halfway between the batch-based approach (typical of the standard model) and the full version of Formulation~\ref{fo:capped}.

\section{Appendix B: Overview of PICEF}

In this section, we briefly overview the Position-Indexed Chain-Edge Formulation (PICEF), the current state-of-the-art approach to optimally clearing traditional batch, single-donor-single-patient kidney exchanges~\cite{Dickerson16:Position}.  PICEF is an IP formulation that uses one binary variable for each cycle in a compatibility graph, but uses binary variables for only the \emph{edges} in chains, not for each chain itself---of which there would be too many to even write down in memory.  The intuition for this representation was drawn from earlier algorithms based on the prize-collecting traveling salesperson problem (PC-TSP)~\cite{Anderson15:Finding}, and on prior ``edge formulation''-style algorithms~\cite{Constantino13:New}.

The innovation of PICEF was the use of position indices on edge variables, which results in polynomial counts of \emph{both} constraints and variables (assuming a constant cycle cap but non-constant chain cap, which is realistic); this is in stark contrast to the \emph{exponential} number of constraints in the PC-TSP-based model.

We first define an index set $\K'(i,j)$, the set of possible positions at which edge $(i,j)$ may occur in a chain in the compatibility graph $G$. Given a maximum chain length $K$ (which can be a constant or allowed to rise with $|V_p|$), for $i, j \in V$ such that $(i, j) \in E$,
\[
\K'(i,j) =
\begin{cases}
    \{1\} & i \in V_n \\
    \{2, \dots, K\} & i \in V_p
\end{cases}.
\]

Any edge leaving a non-directed donor $v \in V_n$ can only be in position \num{1} of a chain, and any
edge leaving a patient-donor vertex in $V_p$ may be in any position up to the chain-length cap
$K$, except \num{1}---that is, it cannot trigger a chain.

For each $(i,j) \in E$ and each $k \in \K'(i,j)$, create variable $y_{ijk}$,
which takes value $1$ if and only if edge $(i,j)$ is selected at position $k$ of
some chain. For each cycle $c$ in $G$ of length up to $L$, define a binary variable $z_c$ to
indicate whether $c$ is used in a packing.   Then we can define the formal PICEF model, shown in Formulation~\ref{form:picef}.

\begin{formulation}[!h]
\rule{\linewidth}{1px}\\[3mm]
\begin{subequations}
{\small
  \begin{align}
    \max && \sum_{(i,j)\in E} \sum_{k\in \K'(i,j)} w_{ij}y_{ijk} && \nonumber\\ 
    && + \sum_{c \in \cycles{}} w_c z_c & \label{eq:picef_obj}\\
    \text{s.t.} && \sum_{j:(j,i)\in E} \sum_{k\in \K'(j,i)} y_{jik} && \nonumber\\
    && + \sum_{\mathclap{c \in \cycles{} : i \text{ appears in } c}} z_c \leq 1 && i \in V_p \label{eq:picef_a} \\
    && \sum_{j:(i,j)\in E} y_{ij1} \leq 1 && i \in V_n \label{eq:picef_b} \\
    && \sum_{\substack{j: (j,i) \in E \wedge \\ k \in \K'(j,i)}} y_{jik} \geq \sum_{j: (i,j) \in E}y_{i,j,k+1} && 
    \begin{aligned}
      & i \in V_p,\\ 
      & k \in \{1, \dots, K-1\}
    \end{aligned}
    \label{eq:picef_c} \\
    && y_{ijk} \in \{0,1\} &&
    \begin{aligned}
      & (i,j) \in E,\\
      & k \in \K'(i,j)
    \end{aligned}
    \label{eq:picef_d} \\
    && z_c \in \{0, 1\} && c \in \cycles{} \label{eq:picef_e}
  \end{align}
}
\end{subequations}
\rule{\linewidth}{1px}
\caption{PICEF formulation.}
\label{form:picef}  
\end{formulation}

Inequality~(\ref{eq:picef_a}) is the capacity constraint for patients: each
patient vertex is involved in at most one chosen cycle or incoming edge of a
chain.  Inequality~(\ref{eq:picef_b}) is the capacity constraint for altruists: each
altruist vertex is involved in at most one chosen outgoing edge.  The flow
inequality~(\ref{eq:picef_c}) ensures that patient-donor pair vertex $i$ has an
outgoing edge at position $k+1$ of a selected chain only if $i$ has an incoming edge at
position $k$; we use an inequality rather than an equality since the final vertex of a chain will have an incoming edge but no outgoing edge.

Code for the PICEF model can be found at the \url{jamestrimble/kidney\_solver} repo on Github.  For more information on this model, and for theoretical and experimental results regarding its correctness and efficacy, we direct the reader to the published paper~\cite{Dickerson16:Position}.

\end{document}